\documentclass[12pt,a4paper]{article}

\usepackage[utf8]{inputenc} 
\usepackage[T1]{fontenc}
\usepackage{url}
\usepackage{ifthen}
\usepackage{cite}
\usepackage[cmex10]{amsmath}

\usepackage[pdftex]{graphicx}
\usepackage{amsmath,amssymb,amsfonts,amsthm}
\usepackage{braket}
\usepackage{array}
\usepackage{comment}
\usepackage{algorithmic,algorithm}
\usepackage{bm}
\usepackage{multirow,array}

\newtheorem{Theorem}{Theorem}
\newtheorem{Definition}[Theorem]{Definition}

\newtheorem{Lemma}[Theorem]{Lemma}

\allowdisplaybreaks[4]

\title{Quantum Deletion Codes Derived From Quantum Reed-Solomon Codes}

\author{
Manabu HAGIWARA
\thanks{
Department of Mathematics and Informatics,
Graduate School of Science,
Chiba University
1-33 Yayoi-cho, Inage-ku, Chiba City,
Chiba Pref., JAPAN, 263-0022.
E-mail: hagiwara@math.s.chiba-u.ac.jp
}
}

\date{2023/06/22}

\begin{document}

\maketitle

\begin{abstract}
This manuscript presents a construction method 
for quantum codes capable of correcting multiple deletion errors. 
By introducing two new alogorithms, the alternating sandwich mapping and the block error locator,
the proposed method reduces deletion error correction 
to erasure error correction.
Unlike previous quantum deletion error-correcting codes,
our approach enables flexible code rates
and eliminates the requirement of knowing the number of deletions.
\end{abstract}

\section{Introduction}
This manuscript discusses the construction of quantum codes 
capable of correcting quantum multiple deletion errors. 
Quantum error-correcting codes have gained significant attention 
in the fields of quantum computing and communication 
due to their essential role in error resilience 
\cite{divincenzo1996fault,gottesman2010introduction,ekert1996quantum}.
In order to improve error resilience in computation and communication, 
various proposals have been developed 
for constructing quantum error-correcting codes, 
primarily by applying classical bit-flip error-correcting codes 
to unitary error correction 
\cite{calderbank1996good,steane1996multiple,gottesman1997stabilizer}.

Error-correction plays a crucial role in various applications
in both classical and quantum information domains
\cite{chang2001reed,howard2006error,surekha2015payment,shor2000simple}.
Recently, there has been growing interest 
within classical coding theory regarding an error model 
called deletion error 
\cite{kulkarni2013nonasymptotic,wachter2017limits,schoeny2017codes}.
Classical deletion error correction was proposed in the 1960s.
New applications were pointed out 50 years later.
This error model holds potential applications 
in error correction for DNA storage \cite{buschmann2013levenshtein} 
and racetrack memory \cite{chee2018coding,mappouras2019greenflag}. 
Additionally, current quantum communication is 
susceptible to deletion errors caused by photon loss during transmission 
\cite{bergmann2016quantum,bergmann2016quantumNOON}. 
Therefore, the development of quantum deletion error-correcting codes 
may be crucial for both present and future applications 
in quantum communication.

In 2020, 
Nakayama discovered a specific example of a quantum error-correcting code 
capable of correcting quantum single deletion errors \cite{nakayama2020first}.
Since then, several papers have been published that provide construction methods 
for quantum deletion error-correcting codes 
\cite{hagiwara2020four,nakayama2020single,shibayama2020new,ouyang2021permutation,shibayama2021permutation,matsumoto2022constructions,shibayama2022equivalence}. 
However, there are some issues that need to be addressed.
Firstly, there is a lack of a method 
that utilizes previously known coding theory 
in the construction of quantum deletion codes. 
Their methods for quantum deletion code construction 
are based on combinatorial or type theory/permutation invariance. 
Secondly, the code rate of currently known quantum deletion error-correcting codes 
is limited and small.
It is desirable to have various choices for the code rate and achieve a high code rate.  
Thirdly, the decoders require the information of
the number of quantum systems for received quantum states.
For example, if no deletion error happens, 
the decoder does not change the received quantum state.
On the other hand, if deletion error happens,
the decoder performs error-correction to the state.
It means that the decoder requires a photon counter 
that does not change quantum states.
Lastly, the currently known quantum deletion codes handle mainly single error-correction.

This manuscript presents a construction method for
quantum deletion error-correcting codes by leveraging quantum Reed-Solomon codes.
The resulting code is capable of multiple deletion errors.
The proposed construction can achieve any desired code rate
under a fixed number of multi-deletions.
It should be noted that the constructed quantum deletion error-correcting codes have a small relative distance,
where the relative distance is the ratio of the minimum distance to the code length.
This small relative distance arises due to 
the necessity of handling codes with large lengths 
while the maximum number of correctable deletion errors is fixed.
Constructing codes with larger relative distances remains a future challenge.
Our decoder does not require a photon counter for received states
under the assumption where the number of deletions is less than or equal to
the previously fixed maximum number of deletions.
\section{Preliminaries}\label{Pre}
This paper assumes the fundamental knowledge of quantum information theory and 
classical/quantum Reed-Solomon codes, particularly described in Sections \ref{ssec:FQIT} and \ref{ssec:CQRSC}.

Throughout this paper,
$n, t, E, K_C, K_D$ and $N$ denote positive integers.
$[n]$ denotes a set of $n$-integers $\{1,2,\ldots ,n\}$.
For $0 \le t \le n$,
$\binom{[n]}{t}$ denotes the set of vectors $(i_1, i_2, \dots, i_t)$
such that $1 \le i_1 < i_2 < \dots < i_t \le n$.
For a set $X$, $\# X$ denotes the cardinality of $X$.

For a positive integer $E$,
$\mathbb{F}_{2^E}$ denotes a field of size $2^E$ 
which is an extended field of the binary field $\mathbb{F}_2 = \{0, 1\}$,
and thus an element of $\mathbb{F}_{2^E}$ can be represented as 
an $E$-bit sequence.

$\bm{0}^{(t)}$ represents the $t$-repetition of the bit $0$, i.e., 
$00 \dots 0 \in \{0,1\}^t$,
and
$\bm{1}^{(t)}$ represents the $t$-repetition of the bit $1$, i.e., 
$11 \dots 1 \in \{0,1\}^t$.

\subsection{Fundamentals of Quantum Information Theory}\label{ssec:FQIT}

For a square matrix $ \tau $ over a complex field $\mathbb{C}$,  $\mathrm{Tr}( \tau )$ denotes
the sum of the diagonal elements of $ \tau $. 
Define $\ket{0},\ket{1}\in \mathbb{C}^2$ as $\ket{0}:=(1,0)^T,\ket{1}:=(0,1)^T$, and 
$\ket{\bm{x}}$ as $\ket{\bm{x}}:=\ket{x_1}\otimes \ket{x_2}\otimes \cdots \otimes \ket{x_n} \in \mathbb{C}^{2 \otimes n}$
for a bit sequence $\bm{x}=x_1x_2\cdots x_n\in \{0,1\}^n$. 
Here $\otimes$ is the tensor product operation, $T$ is the transpose operation, and 
$\mathbb{C}^{2 \otimes n}$ is the $n$th tensor product of $\mathbb{C}^2$, i.e.,
$\mathbb{C}^{2 \otimes n} := (\mathbb{C}^2)^{\otimes n}$.
We may identify $\mathbb{C}^{2 \otimes n}$ with $\mathbb{C}^{2^n}$ as a complex vector space.
We denote by $S(\mathbb{C}^{2\otimes n})$ the set of all density matrices of order $2^n$.
A density matrix is employed to represent a quantum state.
The quantum state $\tau$ that relates to $n$ qubits is represented in an element of
$S(\mathbb{C}^{2\otimes n})$.
Any state $\tau$ is represented in the following form:
\begin{align}
\tau=\sum_{\bm{x},\bm{y}\in \{0,1\}^n} \tau_{\bm{x},\bm{y}}\ket{x_1}\bra{y_1}\otimes \cdots \otimes \ket{x_n}\bra{y_n},
\end{align}
where $\tau_{\bm{x},\bm{y}} \in \mathbb{C}$ and $\bra{a} := \ket{a}^\dagger$,
i.e. $\bra{a}$ is the conjugate transpose of $\ket{a}$.
The quantum state of a subsystem that relates to 
$(n-1)$ qubits is described by the partial trace defined below.

\begin{Definition}[Partial Trace, $\mathrm{Tr}_i$]
Let $i\in [n]$.
Define a function $\mathrm{Tr}_i:S(\mathbb{C}^{2\otimes n})\rightarrow S(\mathbb{C}^{2\otimes (n-1)})$
as
\begin{align*}
\mathrm{Tr}_i( \tau )&:=&
\sum_{\bm{x},\bm{y}\in \{0,1\}^n}
\tau_{\bm{x},\bm{y}}\cdot \mathrm{Tr}(\ket{x_i}\bra{y_i}) \ket{x_1}\bra{y_1}\otimes \\
&&\cdots \otimes \ket{x_{i-1}}\bra{y_{i-1}}\otimes \ket{x_{i+1}}\bra{y_{i+1}}\otimes \\
&&\cdots \otimes \ket{x_n}\bra{y_n},
\end{align*}
where
\begin{align}
\tau=\sum_{\bm{x},\bm{y}\in \{0,1\}^n} \rho_{\bm{x},\bm{y}}\ket{x_1}\bra{y_1}\otimes \cdots \otimes \ket{x_n}\bra{y_n}
\end{align}
and $\mathrm{Tr}( \ket{ x_i } \bra{ y_i } )$ is $1$ if $x_i = y_i$, otherwize $0$.
The map $\mathrm{Tr}_i$ is called the partial trace.
\end{Definition}

\subsection{Classical/Quantum Reed-Solomon Codes}\label{ssec:CQRSC}

Reed-Solomon codes are a class of classical error-correcting codes.
It is known that every Reed-Solomon code is an MDS code,
in other words,
$N - K = d + 1$ holds,
where $N$ is the code length, $K$ is the dimension, and
$d$ is the minimum Hamming distance.

Here we recall an instance of Reed-Solomon code construction.
Assume that $t, E, K_C, K_D,$ and $N$ satisfy
 $t \le K_C \le N-t$, $N - K_C \le K_D$, and $N \le 2^E - 1$.
Let $\alpha$ be a primitive element of $\mathbb{F}_{2^E}$.
Define the $K_D$-by-$N$ matrix 
$H_{D^\perp} := (h_{i,j})_{1 \le i \le K_D, 1 \le j \le N}$ as follows:
\begin{align}
h_{i,j} := \alpha^{(i-1)(j-1)}.
\end{align}
Then $D^{\perp}$ is defined as a linear code over $\mathbb{F}_{2^E}$
with parity-check matrix $H_{D^{\perp}}$,
which is a (shortened) Reed-Solomon code over $\mathbb{F}_{2^E}$.
In other words,
\begin{align}
D^{\perp}
=
\{ \mathbf{x} \in (\mathbb{F}_{2^{E}})^N \mid H_{D^\perp} \mathbf{x}^T = (\mathbf{0}^{(K_D)})^T \}.
\end{align}
The dimension of $D^{\perp}$ is $N - K_D$.
Next, 
we define the $(N-K_C)$-by-$N$ matrix $H_C := (h_{i,j})_{1 \le i \le N-K_C, 1 \le j \le N}$ 
as a submatrix of $H_C$.
We define $C$ as 
a linear code over $\mathbb{F}_{2^E}$ with its parity-check matrix $H_C$,
which is also a (shortened) Reed-Solomon code.
The dimension of $C$ is $K_C$
The minimum Hamming distance of $C$ over $\mathbb{F}_{2^E}$ of $C$ is $N-K_C+1$,
which is greater than $t$ because $K_C \le N-t$.
Therefore, $C$ is a classical $t$-erasure error-correcting code over $\mathbb{F}_{2^E}$.
Since $H_C$ is a submatrix of $H_{D^{\perp}}$,
$C$ includes $D^{\perp}$, i.e. $D^{\perp} \subset C$.

Recall that a quantum Reed-Solomon code $\mathcal{R}$ is a non-binary CSS code
constructed from a pair of classical Reed-Solomon codes $C$ and $D^{\perp}$
that have the same code length, say $N$, and satisfy $D^{\perp} \subset C$.
The quantum code $\mathcal{R}$ is realized as a state of $N$ quantum systems of level $2^{E}$,
i.e., $\mathcal{R} \subset S( \mathbb{C}^{2^E \otimes N})$.
The code length of $\mathcal{R}$ of level $2^E$ is said to be $N$.
The complex dimension of $\mathcal{R}$ is
\begin{align}
2^{E (K_C + K_D - N)}.
\label{eq:qrs_dim}
\end{align}

The quantum minimum distance $d_{\mathcal{R}}$ of level $2^E$
is lower bounded by both the minimum Hamming distance $d_C$ of $C$ and 
the minimum Hamming distance $d_D$ of the dual code of $D^{\perp}$,
i.e., 
\begin{align}
d_{\mathcal{R}} \ge \min\{ d_C, d_D \} = \min\{N - K_C +1, N - K_{D} + 1 \}
\label{eq:minR}
\end{align}
since $d_C = N - K_C + 1$ and $d_D = N - K_D + 1$.

It means that the quantum Reed-Solomon code $\mathcal{R}$  is
capable of quantum $d_{\mathcal{R}}-1$ or less erasure error,
where a quantum $t'$-erasure error is a quantum error $\mathcal{E}_{\bm{i}}$
that transforms states of $t'$ quantum systems of level $2^E$
and announces the $t'$-error position $\bm{i} \in \binom{[N]}{t'}$.
A decoder for a quantum Reed-Solomon code utilize
both the received state and the error position information $\bm{i}$.
Denoting the encoder by 
$\mathrm{Enc}_{\mathcal{R}}
$,
the decoder by 
$\mathrm{Dec}_{\mathcal{R}}
$,
and erasure error operation by $\mathcal{E}_{\bm{i}}$ at position $\bm{i}$,
we have
\begin{align}
\mathrm{Dec}_{\mathcal{R}} 
 \left( \bm{i}, \mathcal{E}_{\bm{i}} \circ \mathrm{Enc}_{\mathcal{R}}( \sigma ) \right)
  = \sigma \label{eq:qrs_ec}
\end{align}
for any $\sigma \in S( \mathbb{C}^{2^E \otimes (K_C + K_D - N) } )$
and any $\bm{i} \in \binom{ [N] }{ t' }$ with $0 \le t' \le d_{\mathcal{R}}-1$.

Since $\mathbb{C}^{2^E \otimes N}$ is isomorphic to $\mathbb{C}^{2 \otimes NE}$
as a complex linear space, 
$\mathcal{R}$ is also realized as a set of states of $NE$ quantum systems of level $2$,
i.e., $\mathcal{R} \subset S( \mathbb{C}^{2^\otimes NE})$.
The codeword is  realized by $NEN$ quantum systems of level $2$,
say $r_1, r_2, \dots, r_{NE}$.
Let us devide the quantum systems into $N$ blocks so that
the $i$th block consists of $E$ quantum systems 
$r_{1 + (i-1)E}, r_{2 + (i-1)E}, \dots r_{E + (i-1)E}$.
In this case, $\mathrm{Dec}_{\mathcal{R}}$ is capable of $t$ or less block erasure.
The index $\bm{i}$ shows the block error positions.

\begin{Definition}[Code Rate]
Let $Q$ be a quantum code such that
$Q$ is a subvector space of $\mathbb{C}^{2 \otimes n}$ as a complex vector space.
If $Q$ is of dimension $M$,
the code rate of $Q$ is defined as $(\log_2 M) / n$. 
\end{Definition}
For example,
the code rate of $\mathcal{R}$ is 
\begin{align}
E (K_C + K_D - N)/ n = (K_C + K_D - N)/ N, \label{eq:qrs_rate}
\end{align}
by Eq.(\ref{eq:qrs_dim}) and $n=NE$.

\section{Quantum Deletion Error}\label{Deletion}
This section provides definitions and observations 
on quantum deletion errors from the perspective of changes in quantum states. 
In line with previous studies (e.g., \cite{nakayama2020first,ouyang2021permutation,shibayama2020new}), 
we define quantum deletion errors using partial trace operations.

\begin{Definition}[Quantum Single/Multi Deletion Error $\mathcal{D}_i$, $\mathcal{D}_{\mathbf{i}}$]
First, we define a quantum single deletion error as a mapping 
from a state $\tau \in S(\mathbb{C}^{2 \otimes n})$ 
to a state $\mathrm{Tr}_i (\tau) \in S(\mathbb{C}^{2 \otimes (n-1)})$, 
where $n$ and $i$ are positive integers with $1 \le i \le n$. 
We denote this error by $\mathcal{D}_i$.

A quantum multi deletion error is a composition of quantum single deletion errors, 
represented as $\mathcal{D}_{i_1} \circ \mathcal{D}_{i_2} \circ \dots \circ \mathcal{D}_{i_t}$, 
for some $\mathbf{i} = (i_1, i_2, \dots, i_t) \in \binom{[n]}{t}$. 
Here, $\circ$ denotes the composition of mappings. 
We denote this error by $\mathcal{D}_{\mathbf{i}}$.
\end{Definition}

We aim to observe quantum deletion errors using quantum systems.
For each $1 \le j \le n$, 
let $p_j$ represent a physical system whose state can be described as 
a two-level quantum state in $S(\mathbb{C}^2)$. 
Examples of physical systems include trapped ions, quantum dots, nitrogen-vacancy centers, photons, and others. 
Their states are commonly referred to as qubits. 
The state of $n$ systems, $p_1, p_2, \dots, p_n$, is represented as an element of $S(\mathbb{C}^{2 \otimes n})$, 
say $\tau$. 
If we fix $1 \le i \le n$ and focus on the remaining $n-1$ systems 
by excluding $p_i$, their state can be described by $\mathrm{Tr}_i (\tau)$. 
As per the definition of deletion errors, this state is represented as $\mathcal{D}_i (\tau)$. 
Similarly, for $\mathbf{i} = (i_1, i_2, \dots, i_t) \in \binom{[n]}{t}$, 
the resulting state after a multi deletion error corresponds 
to the quantum state of the remaining $n-t$ systems, 
represented as $D_{\mathbf{i}}(\tau)$. 
Therefore, deletion errors can be interpreted as errors caused 
by excluding specific physical systems.
An example of excluding systems is loss of photons. 
The loss of photons in quantum communication has been extensively studied [34, 36]. 
Quantum deletion errors serve as a model for errors caused by such losses.

Suppose that $p_i$ is excluded, and the state became $\mathcal{D}_i (\tau)$. 
If we possess $p_i$, we can recover the quantum state 
from $\mathcal{D}_i (\tau)$ to $\tau$ by inserting $p_i$ between $p_{i-1}$ and $p_{i+1}$. 
This can be considered as a deletion error-correction. 
Alternatively, let us introduce another system denoted as $q$, 
and perform a state swap operation between $p_i$ and $q$. 
By inserting $q$ between $p_{i-1}$ and $p_{i+1}$,
the systems are $p_1, p_2, \dots, p_{i-1}, q, p_{i+1}, \dots, p_n$.
Then we can also recover the quantum state from $\mathcal{D}_i (\tau)$ to $\tau$. 
Although the systems are different, the quantum state is the same as the original state $\tau$.
This can be also regarded as a deletion error-correction. 
It should be noted that it is not generally assumed that the excluded system $p_i$ is retained.
Furthermore, knowledge of the exact error position $i$ should not be assumed.

\section{Alternating Sandwich Mapping and Error Locator}
\label{Sec:AsmAndLoc}
This section proposes two algorithms,
Alternating Sandwich Mapping $\mathrm{Asm}$
and Error Locator $\mathrm{Loc}$.
From here, 
the symbol $\mathfrak{o}$ denotes a quantum system of level 2, 
corresponding to the quantum state $\ket{0} \bra{0}$. 
Similarly, $\mathfrak{o}^t$ denotes $t$-$\mathfrak{o}$s,
i,e.,
$\mathfrak{o}, \mathfrak{o}, \dots, \mathfrak{o}$.
In this context, the $\mathfrak{o}^{t}$'s state can be written as 
$\ket{\mathbf{0}^{(t)}} \bra{\mathbf{0}^{(t)}}  \in S(\mathbb{C}^{2 \otimes t})$,
where
$\ket{\mathbf{0}^{(t)}} := 
\ket{0} \otimes \ket{0} \otimes \dots \otimes \ket{0}$. 
Similarly, we use the symbol $\mathfrak{l}^t$ 
to denote a quantum system of level $2^t$, 
whose quantum state is
$\ket{\mathbf{1}^{(t)}} \bra{\mathbf{1}^{(t)}}  \in S(\mathbb{C}^{2 \otimes t})$, where
$\ket{\mathbf{1}^{(t)}} := 
\ket{1} \otimes \ket{1} \otimes \dots \otimes \ket{1}$. 
The notation $\mathfrak{l}^t$ is employed.

\begin{Definition}[Alternating Sandwich Mapping: $\mathrm{Asm}$]
\label{def:asm}
The mapping defined by the following procedure is called 
the alternating sandwich mapping, denoted as $\mathrm{Asm}$.
\begin{enumerate}
\item[Input]: 
A quantum state $\rho \in S(\mathbb{C}^{2 \otimes NE})$,
which is realized by  
$NE$ quantum systems $r_1, r_2, \dots, r_{NE}$ of level 2.
\item[Output]: A quantum state $\tau \in S(\mathbb{C}^{2 \otimes N(E+2t)})$.
\item 
Divide the $NE$ quantum systems into 
$N$ blocks of $E$ systems each, 
denoted as $R_{b} := r_{1 + (b-1)E}, \dots, r_{E + (b-1)E}$ $(1 \le b \le N)$.
Append  
$\mathfrak{o}^t \mathfrak{l}^t $ to the end of each block.
Thus, we have $N(E+2t)$ quantum systems:
$R_1, \mathfrak{o}^t, \mathfrak{l}^t$,
$R_2, \mathfrak{o}^t, \mathfrak{l}^t$,
$\dots$,
$R_N, \mathfrak{o}^t, \mathfrak{l}^t$.
\item 
Output the state $\tau \in S(\mathbb{C}^{2 \otimes N(E+2t)})$ 
of these $N(E+2t)$ quantum systems.
\end{enumerate}
\end{Definition}

\begin{figure}[htbp]
\begin{center}
\includegraphics[width=12cm,bb=0 0 725 362]{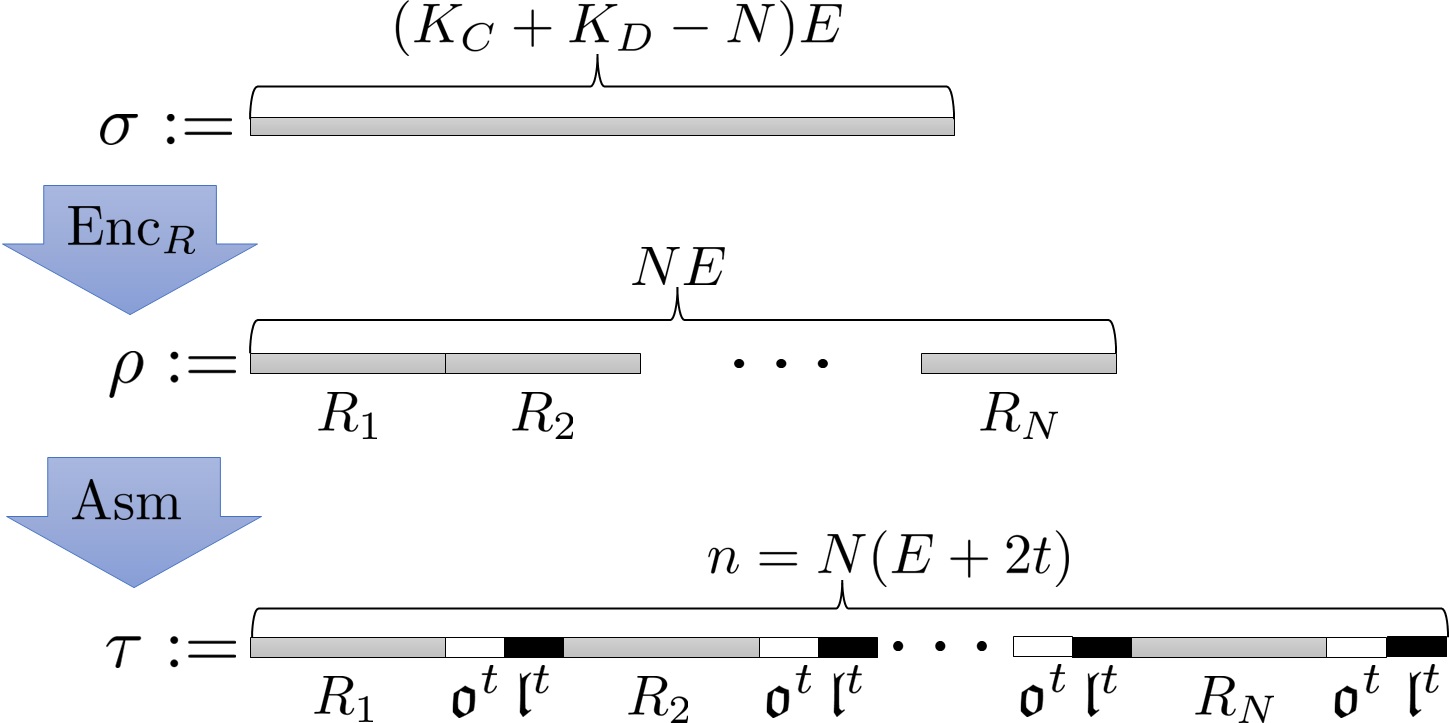} 
\caption{Reed-Solomon Encoder $\mathrm{Enc}_R$ and Alternating Sandwich Mapping $\mathrm{Asm}$}
\label{figure:Enc}
\end{center}
\end{figure}

Consider a case, in Definition \ref{def:asm}, where
 $\rho$ is a pure state given by $\rho = \ket{ \phi } \bra{ \phi }$
with a certain state
\begin{align}
\ket{ \phi } = 
\sum_{ \mathbf{c} \in (\mathbb{F}_{2^{E}})^N  }
  \alpha_{\mathbf{c}} \ket{ c_1 } \ket{ c_2} \dots \ket{ c_N },
\end{align}
where $\mathbf{c} = c_1 c_2 \dots c_N$.
In this case,
$\tau := \mathrm{Asm}( \rho )$ is also a pure state $\ket{ \Phi }\bra{ \Phi }$, and
\begin{align*}
&\ket{ \Phi } = \\
&\sum_{ \mathbf{c} \in (\mathbb{F}_{2^E})^N }
  \alpha_{\mathbf{c}} 
   \ket{ c_1  \mathbf{0}^{(t)} }
   \ket{ \mathbf{1}^{(t)}  c_2  \mathbf{0}^{(t)} }
   \dots
   \ket{ \mathbf{1}^{(t)}  c_N  \mathbf{0}^{(t)} }
   \ket{ \mathbf{1}^{(t)} }.
\end{align*}

\begin{Lemma}\label{lem:key}
Let $n$, $t$, $t'$, and $a$ be positive integers 
satisfying $1 \le t' \le t \le n$ and $a + 2t \le n$.
Let $\tau \in S( \mathbb{C}^{2 \otimes n} )$
be a state represented as
\begin{align}
\tau = \sum_{j} c_j \rho_{1,j} \otimes \rho_{2,j} \otimes \rho_{3,j} \otimes \rho_{4,j},
\end{align}
where $c_j \in \mathbb{C}$,
$\rho_{1,j} \in S( \mathbb{C}^{2 \otimes a} )$,
$\rho_{2,j} = \ket{ \mathbf{0}^{(t)} } \bra{ \mathbf{0}^{(t)} }$,
$\rho_{3,j} = \ket{ \mathbf{1}^{(t)} } \bra{ \mathbf{1}^{(t)} }$,
and
$\rho_{4,j} \in S( \mathbb{C}^{2 \otimes (n-a-2t)} )$.
Assume that
$\tau$ is
realized by $n$ quantum systems $q_1, q_2, \dots, q_n$.

After a $t'$-deletion error occured in the $n$ quantum systems,
the number of deletions in the first $a+t$ quantum systems $q_1, q_2, \dots, q_{a+t}$
is determined by the measurement in the computational basis $\{ \ket{0}, \ket{1} \}$
to quantum systems at positions $a+1$ through $a+t$.
The number of deletions is equal to the Hamming weight of the measurement outcomes.
\end{Lemma}
\begin{proof}
Let us denote the $t'$-deletion error by $\mathbf{i} = (i_1, i_2, \dots, i_{t'})$.
Define integers $t_1, t_2, t_3, t_4$ as follows:
\begin{align*}
t_1 &:= \# \{ i_j \mid 1 \le j \le t', 1 \le i_j \le a \},\\
t_2 &:= \# \{ i_j \mid 1 \le j \le t', a+1 \le i_j \le a+t \},\\
t_3 &:= \# \{ i_j \mid 1 \le j \le t', a+t+1 \le i_j \le a+2t \},\\
t_4 &:= \# \{ i_j \mid 1 \le j \le t', a+2t+1 \le i_j \le n \}.
\end{align*}
In other words, $t_h$ is the number of deletions occured in $\rho_{h,j}$ for $1 \le h \le 4$.
We show that $t_1 + t_2$ is equal to the Hamming weight of the outcomes.

Divide $\mathbf{i}$ into 
four parts $\mathbf{i}_1, \mathbf{i}_2, \mathbf{i}_3, \mathbf{i}_4$,
such that $\mathbf{i}_h \in \binom{[n]}{t_h}$ 
and $\mathbf{i}_h$ is the indices of deletion errors in $\rho_{h,j}$ for $1 \le h \le 4$.
Hence 
\begin{align}
\mathcal{D}_{\mathbf{i}} = 
\mathcal{D}_{\mathbf{i}_1} \circ
\mathcal{D}_{\mathbf{i}_2} \circ
\mathcal{D}_{\mathbf{i}_3} \circ
\mathcal{D}_{\mathbf{i}_4}.
\end{align}
Therefore, we have
\begin{align*}
&\mathcal{D}_{\mathbf{i}} (\tau)\\
=& 
\mathcal{D}_{\mathbf{i}_1} \circ
\mathcal{D}_{\mathbf{i}_2} \circ
\mathcal{D}_{\mathbf{i}_3} \circ
\mathcal{D}_{\mathbf{i}_4} (\tau)\\
=& 
\sum_{j} c_j 
\mathcal{D}_{\mathbf{i}_1} \circ
\mathcal{D}_{\mathbf{i}_2} \circ
\mathcal{D}_{\mathbf{i}_3} \circ
\mathcal{D}_{\mathbf{i}_4} (\rho_{1,j} \otimes \rho_{2,j} \otimes \rho_{3,j} \otimes \rho_{4,j})\\
=& 
\sum_{j} c_j 
\mathcal{D}_{\mathbf{i}_1} \circ
\mathcal{D}_{\mathbf{i}_2} \circ
\mathcal{D}_{\mathbf{i}_3}
(
\rho_{1,j} \otimes \rho_{2,j} \otimes \rho_{3,j} \otimes
\mathcal{D}_{\mathbf{i}_4} ( \rho_{4,j})
)\\
=& 
\sum_{j} c_j 
\mathcal{D}_{\mathbf{i}_1} \circ
\mathcal{D}_{\mathbf{i}_2}
(
\rho_{1,j} \otimes \rho_{2,j} \otimes 
\ket{\mathbf{1}^{t-t_3} } \bra{\mathbf{1}^{t-t_3} }
\otimes \mathcal{D}_{\mathbf{i}_4} ( \rho_{4,j})
)\\
=& 
\sum_{j} c_j 
\mathcal{D}_{\mathbf{i}_1}
(
\rho_{1,j} \otimes 
\ket{\mathbf{0}^{t-t_2} } \bra{\mathbf{0}^{t-t_2} } \otimes 
\ket{\mathbf{1}^{t-t_3} } \bra{\mathbf{1}^{t-t_3} } \otimes 
\mathcal{D}_{\mathbf{i}_4} ( \rho_{4,j})
)\\
=&
\sum_{j} c_j 
\mathcal{D}_{\mathbf{i}_1}
(\rho_{1,j})
\otimes 
\ket{\mathbf{0}^{t-t_2} } \bra{\mathbf{0}^{t-t_2} } \otimes 
\ket{\mathbf{1}^{t-t_3} } \bra{\mathbf{1}^{t-t_3} } \otimes 
\mathcal{D}_{\mathbf{i}_4} ( \rho_{4,j}).
\end{align*}

Note that
$\mathcal{D}_{\mathbf{i}_1} (\rho_{1,j}) \in \mathbb{C}^{2 \otimes (a - t_1)},
\ket{\mathbf{0}^{t-t_2} } \bra{\mathbf{0}^{t-t_2} } \in \mathbb{C}^{2 \otimes (t-t_2)}
$
and 
$\ket{\mathbf{1}^{t-t_3} } \bra{\mathbf{1}^{t-t_3} } \in  \mathbb{C}^{2 \otimes (t-t_3)}$.
Meanwhile, we have
\begin{align}
a - t_1 < a+1
\end{align}
and 
\begin{align}
a+t &\le (a-t_1) + (t - t_2) + (t-t_3) \\
&= a + t + (t - t_1-t_2-t_3).
\end{align}
Therefore, the state associated to the quantum systems from 
$(a+1)$th to $(a+t)$th positions is
$
\ket{\mathbf{0}^{t-(t_1+t_2)} } \bra{\mathbf{0}^{t-(t_1 + t_2)} } \otimes 
\ket{\mathbf{1}^{(t_1 + t_2) } } \bra{\mathbf{1}^{(t_1 + t_2)} }
$.
By the measurement in the computational basis,
the outcomes are $t-(t_1 + t_2)$ zeros and $t_1+t_2$ ones.
In particular, the Hamming weight is $t_1 + t_2$.
\end{proof}

The lemma above motivates us to define the error locator algorithm below.
Recall that the alternating sandwich mapping inserts
$\ket{ \mathbf{0}^{(t)} } \bra{ \mathbf{0}^{(t)} } \otimes \ket{ \mathbf{1}^{(t)} } \bra{ \mathbf{1}^{(t)} }$
between consective blooks of Reed-Solomon codewords.
Figure \ref{figure:Loc} helps to illustrate 
how the algorithm determines the deletion error blocks.

\begin{Definition}[Block Error Locator: $\mathrm{Loc}$]
The mapping defined by the following procedure is called 
the block error locator, denoted as $\mathrm{Loc}$.

Input:
\begin{itemize}
\item A quantum state $\tau'$
which is realized by 
quantum systems $q_1, q_2, \dots, q_{n'}$ of level 2,
where $n' \ge N(E+2t)-t$.
\end{itemize}

Output:
\begin{itemize}
\item $\bm{i'} \in \bigcup_{0 \le t' \le t} \binom{[N]}{t'}$ and
a quantum state $\rho' \in S(\mathbb{C}^{2 \otimes NE})$,
where $\sigma'$ represents a quantum state associated with 
$NE$ quantum systems $y_1, y_2, \dots, y_{NE}$ of level 2.
\end{itemize}

\begin{enumerate}
    \item \textbf{Initialization}:
    \begin{itemize}
        \item[I-1] Initialize $\mathbf{i}' := \emptyset$ and $w_0 := 0$.
        \item[I-2] Append $\mathfrak{l}^{t}$ to the tail of 
        the input quantum systems $q_1, q_2, \dots, q_{n'}$,
        and $\mathbf{Y}$ denotes the systems.
        Hence $\mathbf{Y} = q_1, q_2, \dots, q_{n'} \mathfrak{l}^{t}$.
        \item[I-3] Rename the quantum systems of $\mathbf{Y}$ to $Y_1, Y_2, \dots $.
    \end{itemize}
    \item \textbf{Measurement}:
    \begin{itemize}
        \item[M-1]
        For each $1 \le b \le N$,
        measure each quantum system of the last $t$ quantum systems 
        $Y_{b(E+2t)-t+1}, Y_{b(E+2t)-t+2}, \dots, Y_{b(E+2t)}$
        in the computational basis $\{ \ket{0}, \ket{1} \}$.
        \item[M-2]
        Denote the outcome for $Y_{b(E+2t)-t+i}$ by $s_{b,i} \in \{0, 1\}$.
        Regarding the outcomes as a $t$-bit sequence
        $\mathbf{s}_b := s_{b,1} s_{b,2} \dots s_{b,t} \in \{0, 1\}^t$.
    \end{itemize}
    \item \textbf{Error Detection}:
    \begin{itemize}
        \item[E-1]
        For each $1 \le b \le N$,
        calculate the Hamming weight of $\mathbf{s}_b$, 
        and denote it by $w_b$.
        \item[E-2] If $w_{b+1} = w_b$:
        \begin{itemize}
            \item[E-2-1] Define
             $\mathbf{B}_{b} 
             := Y_{\beta_b+1}, Y_{\beta_b+2}, \dots, Y_{\beta_b+E}$,
             where $\beta_b := (b-1)(E+2t)+w_b+t$.
        \end{itemize}
        \item[E-3] Otherwise, i.e., $w_{b+1} \neq w_b$:
        \begin{itemize}
            \item[E-3-1] Define $E$-quantum systems $\mathbf{B}_{b}$ as $\mathbf{o}^{E}$.
            \item[E-3-2] Update $\mathbf{i}' := \mathbf{i}' \cup \{ b \}$.
        \end{itemize}
    \end{itemize}
    \item \textbf{Terminate}:
    \begin{itemize}
    \item[T-1] Define a quantum state $\rho'$ as the state for
    $EN$ quantum systems
    $\mathbf{B}_1, \mathbf{B}_2, \dots, \mathbf{B}_N$.
    \item[T-2]
    Output $\mathbf{i}'$ and $\rho'$.
    \end{itemize}
\end{enumerate}
\end{Definition}

\begin{figure}[htbp]
\begin{center}
\includegraphics[width=12cm,bb=0 0 769 377]{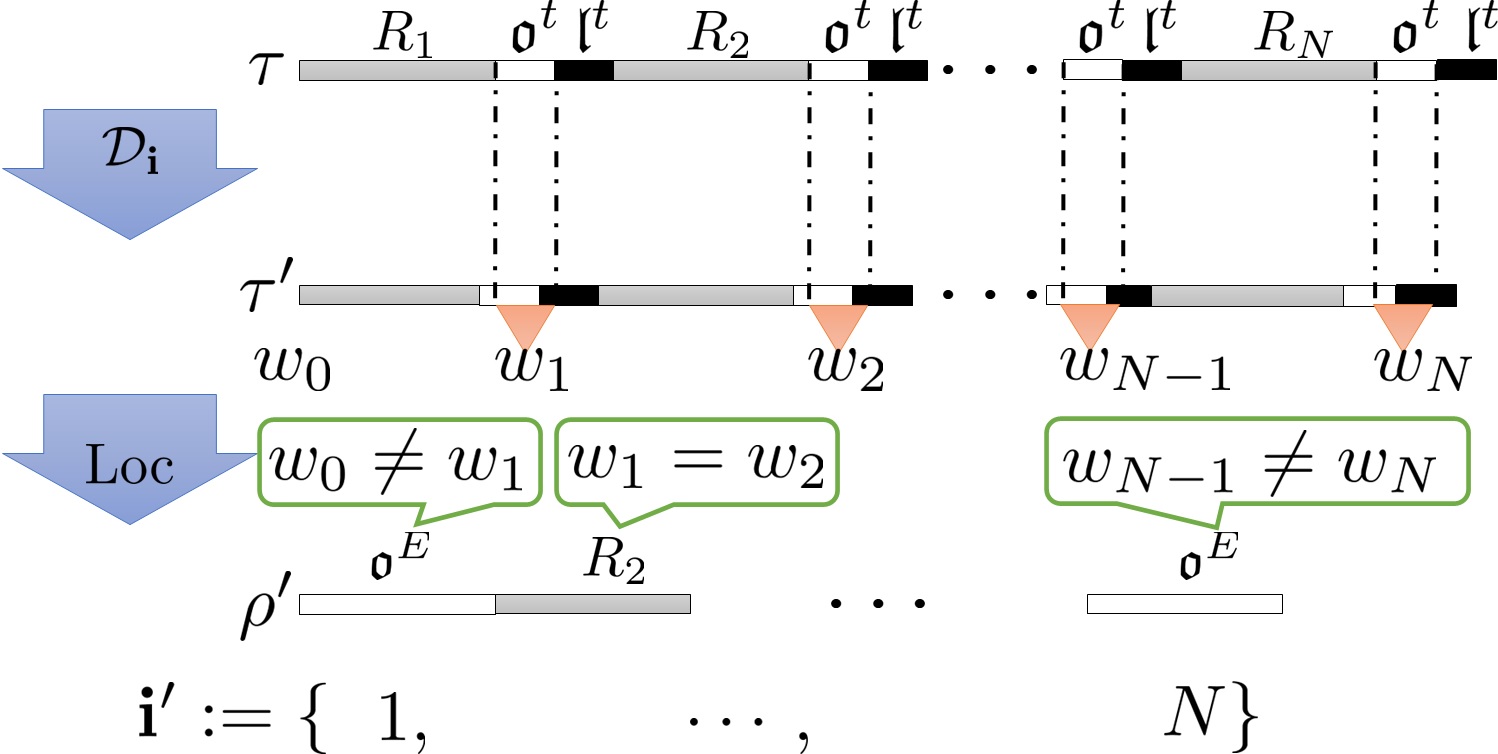} 
\caption{Deletion Error $\mathcal{D}_{\mathbf{i}}$ and Block Error Locator $\mathrm{Loc}$:
$1$ and $N$ belong to $\mathbf{i}$ but $2$ does not.}
\label{figure:Loc}
\end{center}
\end{figure}

\section{Quantum Multi Deletion Error-Correcting Codes from Quantum Reed-Solomon Codes}

\begin{Definition}[$\mathrm{Enc}$, $\mathrm{Dec}$, $\mathcal{Q}$]
Let $\mathrm{Enc}_R $ be
the encoder of a quantum Reed-Solomon code,
and $\mathrm{Dec}_R $ the erasure error-correcting decoder 
for the quantum Reed-Solomon code.
Let $\mathrm{Asm}$ denote the Alternating Sandwich Mapping 
and $\mathrm{Loc}$ the block error LOCator.

We define $\mathrm{Enc}$ and $\mathrm{Dec}$ as follows:
\begin{align}
\mathrm{Enc} &:= \mathrm{Asm} \circ \mathrm{Enc}_R.\\
\mathrm{Dec} &:= \mathrm{Dec}_R \circ \mathrm{Loc}.
\end{align}

Under these conditions, we define our quantum code $\mathcal{Q}$ as follows:
\begin{align}
\mathcal{Q} &:= \mathrm{Enc}( S( \mathbb{C}^{2 \otimes (K_C + K_D - N)E}) ).
\end{align}
\end{Definition}

\begin{Theorem}
For any $\bm{i} \in \binom{[n]}{t}$
and any $\sigma \in S(\mathbb{C}^{2 \otimes NE})$,
we have
\begin{align}
\mathrm{Dec} 
\circ \mathcal{D}_{\bm{i}}
\circ \mathrm{Enc} ( \sigma ) = \sigma.
\end{align}
Hence, $\mathcal{Q}$ is capable of correcting quantum $t$ or less deletion error.
\end{Theorem}
\begin{proof}
By the definition of $\mathrm{Enc}$ and $\mathrm{Dec}$,
\begin{align}
&\mathrm{Dec} 
\circ \mathcal{D}_{\bm{i}}
\circ \mathrm{Enc} (\sigma)\\
=&
\mathrm{Dec}_{R} 
\circ \mathrm{Loc} 
\circ \mathcal{D}_{\bm{i}}
\circ \mathrm{Asm}
\circ \mathrm{Enc}_{R} (\sigma).
\end{align}
Since $\mathrm{Enc}_{R}$ is the encoder of the quantum Reed-Solomon code,
$\mathrm{Enc}_{R} (\sigma)$ is a codeword.
Let $r_1, r_2, \dots, r_{NE}$ be quantum systems of level 2
such that their quantum state is $\mathrm{Enc}_{R} (\sigma)$.
Then the quantum system after $\mathrm{Asm}$
consists of $N$ blocks,
and the $b$th block is
\begin{align}
r_{1 + (b-1)E}, r_{2 + (b-1)E}, \dots, r_{E + (b-1)E}, 
\mathfrak{o}^t, \mathfrak{l}^t,
\end{align}
where $1 \le b \le N$.

By Lemma \ref{lem:key}, for $b > 1$,
the deletion error occured after
$\mathfrak{o}^{t}$ at $(b-1)$th block
and before $\mathfrak{l}^{t}$ at $b$th block
if and only if 
$w_{b-1} = w_b$.
Similarly, 
By Lemma \ref{lem:key}, for $b = 1$,
the deletion error occured among
$r_{1 }, r_{2 }, \dots, r_{E }, \mathfrak{o}^t$
if and only if 
$w_1 = w_0 = 0$.
Hence the block positions of quantum deletions are determined
at the Error Detection step of the block error locator.
Therefore, 
if we denote the output of $\mathrm{Loc}$ as $(\mathbf{i}', \rho')$,
we have
\begin{align}
\rho' = \mathcal{E}_{\mathbf{i}'} \circ \mathrm{Enc} (\sigma).
\end{align}
In particular, 
the quantum state of the error position block is 
$\ket{\mathbf{0}^{(E)} } \bra{\mathbf{0}^{(E)} } $.

From the above, we have
\begin{align}
&
\mathrm{Dec}_{R} 
\circ \mathrm{Loc} 
\circ \mathcal{D}_{\bm{i}}
\circ \mathrm{Asm}
\circ \mathrm{Enc}_{R} (\sigma)\\
=&
\mathrm{Dec}_{R} (\mathbf{i}', \rho')\\
=&
\mathrm{Dec}_{R} (\mathbf{i}', \mathcal{E}_{\mathbf{i}'} \circ \mathrm{Enc} (\sigma)).
\end{align}
By Equantion (\ref{eq:qrs_ec}),
it is equal to $\sigma$.
\end{proof}

Since alternating sandwich mapping is a simple method, 
the author thought it would not be surprising 
if there were previous studies in classical coding theory.
The author reviewed numerous papers on deletion codes, 
including several survey papers (e.g., \cite{mitzenmacher2008survey,mercier2010survey}), 
but did not find any similar ideas.
The author also consulted researchers about deletion codes, 
but no one was familiar with this idea.
A related idea is that of a marker code,
a deletion error-correcting code constructed by
inserting specific bit sequences \cite{ratzer2005marker}.
The inserted bit sequences are different from 
those inserted by the alternating sandwich mapping.
Therefore, alternating sandwich mapping 
may represent a new idea in coding theory,
at least quantum coding theory.
In considering quantum insertion-deletion channels, 
Leahy mentioned the insertion of classical bit strings
\cite{leahy2019quantum}. 
The specific bit string used was $100\dots0$,
where the string consists of consecutive zeros
followed by a single one. 
However, this method fails to correct errors 
if the leading $1$ is deleted. 
Furthermore, 
Leahy did not employ 
non-binary quantum erasure-correcting codes
from Reed-Solomon codes.


Finally, we discuss the achievable code rate of the constructed codes.
The utilization of quantum Reed-Solomon codes 
enables our codes to achieve a flexible code rate.

\begin{Theorem}\label{thm:highRateExistence}
For any real number $0 \le \gamma \le 1$,
there exist quantum deletion error-correcting codes, 
denoted as $\mathcal{Q}_1, \mathcal{Q}_2, \dots$,
capable of correcting $t$ or less deletion errors,
such that their code rates converge to $\gamma$,
i.e., $\lim_{n \rightarrow \infty} \gamma_n = \gamma$,
where $\gamma_n$ represents the code rate of $\mathcal{Q}_n$.
\end{Theorem}
\begin{proof}
For each positive integer $E$ satisfying $2^E - 1  > 2t$,
choose Reed-Solomon codes $C$ and $D^{\perp}$ 
with a length $N := 2^{E}-1$,
such that
$K_{D} := N - t$ and $K_{C} := \lfloor \gamma N \rfloor $,
where $\lfloor a \rfloor$ is the largest integer not exceeding $a$.

The dimension of $\mathcal{Q}_E$ is equal to that of $\mathcal{R}$.
Therefore, it is $2^{(K_C + K_D - N) E} = 2^{(\lfloor \gamma N \rfloor - t) E}$.
The code length of $\mathcal{Q}_E$ is $N (E + 2t) $.
Hence, the code rate of $\mathcal{Q}_E$ is
\begin{align}
\gamma_{\mathcal{Q}_E} =
\frac{(\lfloor \gamma N \rfloor - t) E}{N (E+2t)}.
\end{align}
By the definition of $\lfloor \; \rfloor$,
\begin{align}
\frac{( \gamma N - 1 - t)E}{N(E+2t)}
\le 
\gamma_{\mathcal{Q}_E}
\le
\frac{( \gamma N  - t) E}{N(E+2t)}.
\label{eq:qrateConv}
\end{align}

Recall that $N$ is chosen as $2^{E}-1$.
As $E$ approaches infinity,
$N$ also tends to infinity.
Moreover, both sides of (\ref{eq:qrateConv}) converge to $\gamma$.
Hence
$\gamma_{\mathcal{Q}_E}$
also converges to $\gamma$.
\end{proof}

\section{Conclusion}\label{sec:concRem}

This paper proposed a construction method 
for quantum deletion error-correcting codes.
The idea is to reduce the deletion error correction 
to the erasure error correction of quantum Reed-Solomon codes 
by combining the proposed alternating sandwich mapping 
and the proposed block error locator.
This approach allows us 
to leverage the properties of classical Reed-Solomon codes, 
enabling flexible code rate design.
Furthermore, while previous studies implicitly assumed that 
the number of deletion errors is known for error correction, 
our proposed decoder does not require such an assumption.
However, this study did not achieve an arbitrary relative distance. 
This remains a topic for future research.
What are the potential applications of quantum deletion error-correcting codes?
Quantum erasure error-correcting codes have been applied in information security 
as components of quantum secret sharing protocols.
If a protocol based on quantum deletion error-correcting codes is constructed, 
what advantages would it offer?
Research on quantum deletion error-correcting codes is still in its early stages.
The author hopes that in the future, 
new and efficient codes will be proposed, 
applications will be discovered, 
and implementations will lead to the advancement of the theory.


\section*{Acknowledgments}
This paper is partially supported by
KAKENHI 21H03393.
We would like to express our gratitude to J.B. Nation, Ellen Hughes from University of Hawai'i at Manoa, 
and Justin Kong, Austin Anderson from Kapi'olani Community College 
for providing valuable discussion opportunities.

\bibliographystyle{ieicetr}
\bibliography{reference}

\end{document}